\documentclass{article}
\usepackage[utf8]{inputenc}

\usepackage[ruled,vlined]{algorithm2e}
\usepackage[utf8]{inputenc}
\usepackage[T1]{fontenc}
\usepackage{latexsym,exscale,amssymb,amsmath}

\usepackage{svg}
\usepackage{blindtext}
\usepackage{bbm}

\usepackage{hyperref}
\usepackage{lmodern}
\usepackage{amssymb,bbold}
\usepackage{mathtools}
\usepackage{amssymb} % that is for real numbers etc.
\usepackage{amsmath}
\usepackage{amsthm}   % enables to number all with 1 single counter
\usepackage{amsfonts}
\usepackage[shortlabels]{enumitem}
\usepackage{natbib}
\usepackage[textfont=it]{caption}
\usepackage[textfont=it]{subcaption}
\usepackage{xcolor}
\usepackage{float}
\usepackage{tikz}
\newcommand{\iid}{\operatorname{\stackrel{i.i.d.}{\sim}}}

\DeclareMathOperator*{\argmin}{\mbox{\rm argmin}}

\newtheorem{Def}{Definition}
\newtheorem{Cor}[Def]{Corollary}
\newtheorem{Th}[Def]{Theorem}
\newtheorem{Ex}[Def]{Example}

\newtheorem{Lem}[Def]{Lemma}

\newcommand{\cS}{\mathcal{S}}

\newcommand{\Prb}{\mathbb P}

\newcommand{\inas}{\operatorname{\stackrel{a.s.}{\to}}}

\newcommand{\EE}{\mathbb{E}}

\newcommand{\NN}{\mathbb{N}}

\newcommand{\R}{\mathbb{R}}
\newcommand{\Sp}{\mathcal{S}}

\newcommand{\N}{\mathbb{N}}
\newcommand{\m}{{\mu}}
\newcommand{\dis}{\mathbf{d}}

\newcommand{\bo}{\mathbf{0}}
\newcommand{\E}{\mathbb{E}}
\newcommand{\mycomment}[1]{}

\newcommand{\md}{\mathbf{m}_n}

\newcommand\DirectBibliography[2][Chicago]{%
  \begingroup
    \let\clearpage\relax
    \vspace{2em}
    \bibliographystyle{#1}
    \bibliography{#2}
  \endgroup
  \clearpage
}

\begin{document}
% \pagenumbering{gobble} 

\title{Exploring Uniform Finite Sample Stickiness} 
\author{Susanne Ulmer\footnote{Georg-August-Universit\"at at G\"ottingen, Germany, Felix-Bernstein-Institute 
for Mathematical Statistics in the Biosciences,
% \url{susanne.ulmer@stud.uni-goettingen.de}, \url{do.tranvan@uni-goettingen.de}, \url{Stephan.Huckemann@mathematik.uni-goettingen.de}
 }, 
Do Tran Van$^*$ \and Stephan F. Huckemann$^*$}

\maketitle

\begin{abstract}
It is well known, that Fréchet means on non-Euclidean spaces may exhibit nonstandard asymptotic rates depending on curvature. Even for distributions featuring standard asymptotic rates, there are non-Euclidean effects, altering finite sampling rates up to considerable sample sizes. These effects can be measured by the variance modulation function proposed by \cite{Pennec2019}. Among others, in view of statistical inference, it is important to bound this function on intervals of sampling sizes. In a first step into this direction, for the special case of a K-spider %-- a 3-spider is phylogenetic BHV space \citep{BilleraHolmesVogtmann2001} with 3 leaves -- 
we give such an interval, based only on folded moments and total probabilities of spider legs and illustrate the method by simulations.
 \end{abstract}

\section{Introduction to Stickiness and Finite Sample Stickiness}

Data analysis has become an integral part of science due to the growing amount of data in almost every research field. This includes a plethora of data objects that do not take values in Euclidean spaces, but rather in a non-manifold, stratified space. For statistical analysis in such spaces, it is therefore necessary to develop probabilistic concepts. %With a silently underlying probability space $(\Omega,\cA,\Prb)$,
\cite{F48} was one of the first to generalize the concept of an expected value to a random variable $X$ on an arbitrary metric space $(Q,\dis)$ as an minimizer of the expected squared distance:
\begin{equation} \label{frechetmeanforgeneralmetricspaces}
\m = \argmin_{p \in Q}\EE[d(X,p)^2]\,,
%\int \dis^2(X,p)\,d\Prb\,,
\end{equation} 
nowadays called a \emph{Fr\'echet mean} in his honor. Accordingly for a sample $X_1,\ldots, X_n$ $ \iid X$, its Fr\'echet mean is given by
\begin{equation} \label{frechetsamplemeanforgeneralmetricspaces}
\m_n = \argmin_{p \in Q} \sum_{j = 1}^{n} \dis^2(X_j,p)\,.
\end{equation}
While on general space, these means can be empty or set valued, on \emph{Hadamard spaces}, i.e. complete spaces of \emph{global nonpositive curvature} (NPC), due to completeness, these means exists under very general conditions, and due to simple connectedness and nonpositive curvatures, they are unique, e.g. \cite{Sturm2003}, just as their Euclidean kin. They also share a law of strong numbers, i.e. that
$$ \m_n \inas \m \,.$$ 
In contrast, however, their asymptotic distribution is  often not normal, even worse, for some random variables their mean may be on a singular point stratum, and there may be a random sample size $N\in \NN$ such that
$$ \m_n = \m\mbox{ for all } n\geq N\,.$$
This phenomenon has been called \emph{stickiness} by \cite{HHMMN13}. It  puts an end to statistical inference based on asymptotic fluctuation. While nonsticky means of random variables seem to feature the same asymptotic rate, as the Euclidean expected value, namely $1/\sqrt{n}$, it has been noted by \cite{huckemann2020data} that for rather large sample sizes, the rates appear to be larger. This contribution is the first to systematically investigate this effect of \emph{finite sample stickiness} on stratified spaces and we do this here for the model space of the $K$-spider introduced below. This effect is in some sense complementary to the effect of \emph{finite sample smeariness}, where finite sample rates are smaller than $1/\sqrt{n}$.
recently discovered by \cite{HundrieserEltznerHuckemann2020}.
%and investigated by \cite{DHH-GSI21,EHH-GSI21,eltzner2023diffusion}. 

\begin{Def} With the above notation, assuming an existing Fr\'echet function  $F(x) = \EE[\dis(X,x)^2] < \infty $ for all $x\in Q$, with existing Fr\'echet mean $\m$,
\begin{align} \label{variancemodulation}
\md = \frac{n \E \left[\dis^2(\m_n,\m) \right]}{\E \left[\dis^2(X,\m) \right]},
\end{align}
is the \emph{variance modulation} for sample size $n$ (see \cite{Pennec2019}), %\emph{variance modulation} (see \cite{}) 
or simply \emph{modulation}.
\end{Def}

If $(Q,\dis)$ is Euclidean, then $\md=1$ for all $n\in \NN$, \emph{smeariness} governs the cases $\md \to \infty$, see \cite{HH15,EltznereHuckemann2019,eltzner2022geometrical}, finite sample smeariness the cases $1 < \md$, cf. %\cite{HundrieserEltznerHuckemann2020, 
\cite{DHH-GSI21, eltzner2021finite,eltzner2023diffusion}, stickiness the case that $\mu_n = \mu$ a.s. for $n>N$ with a finite random sample size $N$, see \cite{HHMMN13,H_Mattingly_Miller_Nolen2015, BardenLeOwen2013,Barden2018limiting}, and \emph{finite sample stickiness} %\citep{huckemann2020data} 
the case that
$$ 0< \md < 1\mbox{ for nonsticky } \mu\,.$$
\begin{Def}
\label{definitionfssKspider} For a nonsticky mean, if there are integers $l \in \N_{\geq 2}$,  and $N\in N$ and $0<\rho < 1$ such that
\begin{align*}
0 < \md < 1- \rho
\end{align*}
for all $n \in \{N,N+1,..., N^l \}$ then $X$ is called \emph{finite sample sticky of level $\rho \in (0,1)$, with scale $l$ and basis $N$} .
\end{Def}

We note that \cite{Pennec2019} has shown that finite sample stickiness affects all affine connection manifolds with constant negative sectional curvature.
%We note that finite sample stickiness affects all Hadamard manifolds when data samples cover negative curvature regions, see (Ulmer: Pennec showed finite sample stickiness for Riemannian and affine connection manifolds. He gives explicit convergence bounds for isotropic distributions on manifolds with constant sectional curvature). 
We conjecture that this is also the case for general Hadamard spaces. 

A very prominent example of a nonmanifold Hadamard space is given by the \emph{BHV tree spaces} introduced by  \cite{BilleraHolmesVogtmann2001} modeling phylogenetic descendance trees. %building on the work of \cite{darwin1871origin,haeckel1870naturliche,Felsenstein1978}. 
For a fixed number of species or taxa the BHV-space models all  different tree topologies, where within each topology, lengths of internal edges reflect evolutionary mutation from unknown ancestors. For three taxa, there are three topologies featuring nonzero internal edges and a fourth one, the \emph{star tree} featuring no internal edge. The corresponding BHV space thus carries the structure of a 3-spider as depicted in Figure \ref{fig:BHV-3}. 
For illustration of argument, in this contribution we consider $K$-spiders.

\begin{figure}
    \centering
\begin{tikzpicture}    
%Three-Spider

\draw[orange, thick] (0,0) -- ({1*cos(0)}, {1*sin(0)});
\filldraw[orange] (1,0) circle (2pt);
\draw[gray, thick] ({1*cos(0)}, {1*sin(0)}) -- ({2*cos(0)}, {2*sin(0)});
\draw[gray, thick] (0,0) -- ({2*cos(120)}, {2*sin(120)});
\draw[gray, thick] ({1.5*cos(240)}, {1.5*sin(240)}) -- ({2*cos(240)}, {2*sin(240)});
\filldraw[black] ({2*cos(0)}, {2*sin(0)}) circle (0pt) node[right](13_02){13/02};
\filldraw[black] ({2*cos(120)}, {2*sin(120)}) circle (0pt) node[above](12_03){12/03};
\filldraw[black] ({2*cos(240)}, {2*sin(240)}) circle (0pt) node[below](23_01){23/01};

\draw[blue, thick] (0,0) -- ({0.5*cos(120)}, {0.5*sin(120)});
\filldraw[blue] ({0.5*cos(120)}, {0.5*sin(120)}) circle (2pt);

\draw[green, thick] (0,0) -- ({1.5*cos(240)}, {1.5*sin(240)});
\filldraw[green] ({1.5*cos(240)}, {1.5*sin(240)}) circle (2pt);

%tree 13/02
\draw[orange, thin] (5,-0.5) -- (4.5, -1);
\draw[black, thin] (4.5,-1) -- (4, -1.5);
\draw[black, thin] (4.5,-1) -- (5, -1.5);
\draw[black, thin] (5,-0.5) -- (6, -1.5);
\draw[black, thin] (5,-0.5) -- (5,-0.25);

\filldraw[black] (5, -0.25) circle (0pt) node[above](){0};

\filldraw[black] (5, -1.5) circle (0pt) node[below](){3};
\filldraw[black] (4, -1.5) circle (0pt) node[below](){1};
\filldraw[black] (6, -1.5) circle (0pt) node[below](){2};

%arrow

\draw[orange, dashed, <-> ] plot[smooth] coordinates {(1.5,-0.25) (2,-0.5) (3,-0.7) (4.1, -0.72)};
%startree 
\filldraw[black] (0.5,1) circle (0pt) node[below](){1};
\filldraw[black] (1.5,1) circle (0pt) node[below](){2};
\filldraw[black] (2.5,1) circle (0pt) node[below](){3};
\filldraw[black] (1.5,2.25) circle (0pt) node[above](){0};

\draw[black, thin] (0.5,1) -- (1.5,2);
\draw[black, thin] (1.5,1) -- (1.5,2);
\draw[black, thin] (2.5,1) -- (1.5,2);
\draw[black, thin] (1.5,2) -- (1.5,2.25);

\filldraw[gray] (0,0) circle (2pt);
\draw[gray,dashed, <->] plot[smooth] coordinates {(0.1,0.1) (0.4,0.2) (0.7,0.3) (0.8, 0.4)};

%tree 12/03

\filldraw[black] (-2, 0.5) circle(0pt) node[below](){2};
\filldraw[black] (-3.5,0.5) circle(0pt) node[below](){1};
\filldraw[black] (-4, 0.5) circle(0pt) node[below](){3};
\filldraw[black] (-3, 1.75) circle (0pt) node[above](){0};

\draw[black, thin] (-2,0.5 ) -- (-2.75,1.25);
\draw[blue, thin] (-2.75,1.25) -- (-3,1.5);
\draw[black,thin] (-2.75,1.25) -- (-3.5,0.5);
\draw[black,thin] (-3,1.5)--(-4,0.5);
\draw[black,thin] (-3,1.5)--(-3,1.75);

%arrow
\draw[blue,dashed, <->] plot[smooth] coordinates {(-2,0.8) (-1.2, 0.75)({0.5*cos(120)-0.2}, {0.5*sin(120)})};

%tree 23/01 

\draw[black,thin] (0,-2)--(0.25,-1.75);
\draw[black,thin] (0.25,-1.75) --(0.5,-2);

\draw[green,thin] (0.25,-1.75) --(1,-1);

\draw[black,thin] (2,-2) --(1,-1);
\draw[black,thin] (1,-0.75) --(1,-1);

\filldraw[black] (0,-2) circle(0pt) node[below](){2};
\filldraw[black] (2,-2) circle(0pt) node[below](){1};
\filldraw[black] (0.5,-2) circle(0pt) node[below](){3};
\filldraw[black] (1,-0.75) circle (0pt) node[above](){0};

%arrow
\draw[green, dashed, <->] plot[smooth] coordinates {({1*cos(240)+0.2}, {1*sin(240)}) (0.2, -1.) (0.4, -1.2)};

%label three spider
\filldraw (-3,-1.5) circle(0pt) node[minimum size=3cm] {$\mathcal{S}_3$};
\end{tikzpicture}
    \caption{Four different phylogenetic descendance trees for three taxa featuring one or none internal edge, modeled on the 3-Spider $\cS_3$.}
    \label{fig:BHV-3}
\end{figure}
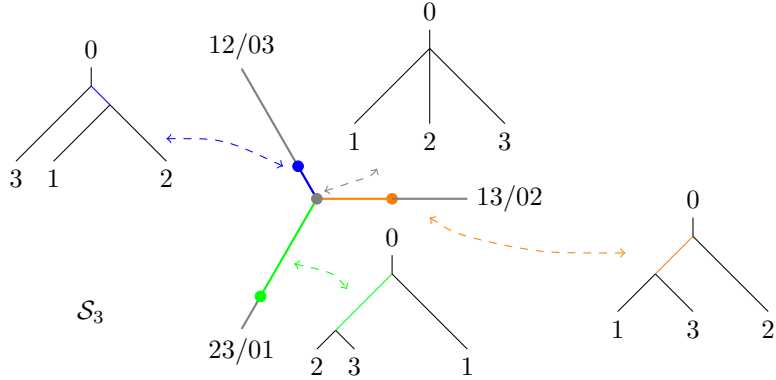

%In the following, we always assume the $X$ is a random variable on a metric space $Q,\dis)$ with existing Fr\'echet function $F(x) = \EE[\dis(X,x)^2] < \infty $ for all $x\in Q$.

%\begin{Def} With the above notation,
%\begin{align} \label{variancemodulation}
%\md = \frac{n \E \left[\dis^2(\m_n,\m) \right]}{\E \left[\dis^2(X,\m) \right]},
%\end{align}
%is the \emph{variance modulation} for sample size $n$ (see \cite{Pennec2019}), %%\emph{variance modulation} (see \cite{}) 
%or simply \emph{modulation}.
%\end{Def}

%If $(Q,d)$ is Euclidean, then $\md=1$ for all $n\in \NN$, \emph{smeariness} governs the cases $\md \to \infty$, see \cite{HH15,EltznerHuckemann18}, finite sample smeariness the cases $1 < \md$, cf. \cite{HundrieserEltznerHuckemann2020, DHH-GSI21, EHH-GSI21}, stickiness the case that $\mu_n = \mu$ a.s. for $n>N$ with a finite random sample size $N$, see \cite{HHMMN13,H_Mattingly_Miller_Nolen2015, BardenLeOwen2013,Barden2018limiting}, and \emph{finite sample stickiness} \citep{huckemann2020data} the case that
%$$ 0< \md < 1\mbox{ for nonsticky } \mu\,.$$
%\begin{Def}
%\label{definitionfssKspider} For a nonsticky mean, if there are integers $l \in \N_{\geq 2}$,  and $N\in N$ and $0<\rho < 1$ such that
%\begin{align*}
%0 < \md < 1- \rho
%\end{align*}
%for all $n \in \{N,N+1,..., N^l \}$ then $X$ is called \emph{finite sample sticky of level $\rho \in (0,1)$, with scale $l$ and basis $N$} .
%\end{Def}

% \vspace*{-0.8cm}

\section{Model Space: The K-Spider}
%The $K$-spider for $K = 3$, defined below, coincides with the BHV-tree space for three leaves. 
The following has been taken from \cite{HHMMN13}.
\begin{Def} \label{definitionspider}
For $3\leq K \in \NN$ the $K$\emph{-spider} $\cS_K$ is the space
\begin{equation*}
\Sp_K = [0, \infty ) \times \{1,2,...,K\} / \sim
\end{equation*}
where for $i,k \in \{1,2,..., K\}$ and $x\geq 0$, $(x,i) \sim (x,k)$  if $k=i$ or $x=0$. The equivalence class of $(0,1)$ is identified with the \emph{origin} $\bo$, so that $\cS_K = \{\bo\}\cup \bigcup_{k=1}^K L_k$ with the positive half-line $L_k := (0, \infty) \times \{k\}$ called the \emph{$k$-th leg}. Further, for any $k \in \{1,2,...,K\}$ the map
\begin{align*}
F_{k}: \Sp_K &\to \mathbb{R},\\ (x,i) &\mapsto \begin{cases} x & if \quad  i= k,  \\ -x & else \end{cases}
\end{align*} 
is called the \textit{$k$-th folding map of $\Sp_K$}.

The first two \emph{folded moments on the $k$-th leg} of a random variable $X$ are
\begin{equation*}
m_k = \EE[F_k(X)],\quad \sigma^2_k =  \EE[(F_k(X)-m_k)^2]
%m_k = \int_{\R} x \fmu_k(dx).
\end{equation*}
in population version and the first folded moment in sample version for random variables $X_1,\ldots,X_n$ is
\begin{align*}
\eta_{n,k} = %\int_{\R} x \fmu_k(dx) = 
\frac{1}{n} \sum_{j = 1}^n F_k(X_j)\,.
\end{align*}
%is the \emph{$k$-th folded average}.

We say the $X$ is \emph{nondegnerate} if $\Prb\{X \in L_k\} >0$ for at least three different $k \in \{1,\ldots,K\}$.

\end{Def}

\begin{figure} [h]
\centering
\begin{tikzpicture} 

\draw[orange, thick] (0,0) -- node[right]{$2$}(0.6,1.8);
\draw[blue, thick] (0,0) -- node[below]{$3$}(-1.8,0);
\draw[green, thick] (0,0) -- node[left]{$1$} (1.05,-1.5);

\draw[black, thin, ->] (1.5,0) --  node[above]{$F_1$}(3,0);
\draw[orange, thick] (3.5,0) -- (5,0);
\draw[dashed, blue, thick] (3.5,0) -- (5,0);
\draw[green, thick] (5,0) -- (7,0);

\filldraw[gray] (0,0) circle (1pt) node[right]{  $\sim \{0 \}$};

\filldraw[gray] (5,0) circle (1pt) node[above]{$0$};

\end{tikzpicture}
\caption{The first folding map $F_1$ on the Three-Spider $\Sp_3$. The leg labelled $1$ is mapped to the positive real line. The second and third leg are folded and then mapped to the negative half line.}
\end{figure}
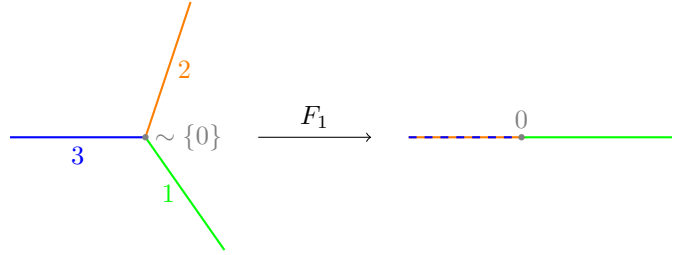

\begin{Lem} \label{folded_average_related_to_folded_sample_mean}
For a sample of size $n$ of a nondegenerate random variable, we have for every $1\leq k\leq K$ that
\begin{eqnarray}\label{eq:lem8-0}\eta_{n,k} > 0 &\Leftrightarrow& \mu_n \in L_k~\Leftrightarrow~F_k(\mu_n) > 0\end{eqnarray}
whereas 
$$\eta_{n,k} = 0 \Rightarrow \mu_n = \bo \in \cS_K \Rightarrow \eta_{n,k} \leq 0\,.$$
In particular:  
\begin{eqnarray}\label{eq:lem8-1}
\eta_{n,k} \geq 0 &\Leftrightarrow& F_k(\m_n) = \eta_{n,k}
\end{eqnarray}
and 
\begin{eqnarray}\label{eq:lem8-2}
\eta_{n,k} < 0 &\Leftrightarrow &\eta_{n,k} < F_k(\m_n)\,.\end{eqnarray}
%If $\eta_{n,k} \geq 0$, then $F_k(\m_n) = \eta_{n,k}$. On the other hand if $\m_n \in H_k$, then $F_k(\m_n) = \eta_{n,k} > 0$ as well.
\end{Lem}

\begin{proof} All but the last assertion are from
\citet[Lemma 3.3]{HHMMN13}. The last is also in the proof of \citet[Lemma 4.1.4]{HE24Foundations}, we prove it here for convenience. 

Letting $h_{n,i} = \frac{1}{n} \sum_{X_j \in L_i} F_i(X_j)$ for $i\in \{1,\ldots,K\}$ we have 
\begin{eqnarray}\label{eq:eta} 
\eta_{n,k} &=& h_{n,k}- \sum_{\substack{ i\neq k\\ 1\leq i\leq K}} h_{n,i}\,.
\end{eqnarray}
If $\eta_{n,k} < 0$ then there must be $k\neq k' \in\{1,\ldots,K\}$ with $\mu_n \in L_{k'}$ and thus
\begin{eqnarray*}
0~< ~-F_k(\mu_n) ~=~ F_{k'}(\mu_n)~=~\eta_{n,k'} &=& h_{n,k'} - h_{n,k}- \sum_{\substack{ i\neq k,k'\\ 1\leq i\leq K}} h_{n,i}\\& =& - \eta_{n,k} - 2\sum_{\substack{ i\neq k,k'\\ 1\leq i\leq K}} h_{n,i} ~<~ -\eta_{n,k}
\end{eqnarray*}
as asserted, due to nondegenaracy.
\qed\end{proof}

The case, $m_k < F_k(\mu)$ for all $k \in \{1,\ldots,K\}$ governs stickiness, cf. \cite{HHMMN13}, and the case of $\eta_{n,k} < F_k(\mu_n) \leq 0 <  F_k(\mu) =m_k$ for some $k\in \{1,\ldots,K\}$ governs finite sample stickiness, as detailed below. For this reason, we consider the following.
%in case of $m_k >0$ we consider
%$$ Y_n := \vert \eta_{n,k}- m_k\vert^2 - \vert F_k(\m_n)- m_k\vert^2 \,.$$ 

\begin{Cor}\label{cor:folded_average}
In case of $m_k >0$, % and $\sigma_k = \EE[|\mu_k -F_k(X)|^2]$, 
we have
$\vert F_k(\m_n)- m_k\vert^2 \le \vert \eta_{n,k} - m_k\vert^2$ and $n\EE [\dis^2(\mu_n,\mu)] = n\EE [\vert F_k(\m_n)- m_k\vert^2] \le \sigma_k^2 \,.$ 
%$$ Y_n \geq 0\mbox{ and } \frac{n\EE [|Y_n|]}{\sigma_k^2} =1 -\fm_n\,.$$ 
\end{Cor}

\begin{proof}
If $\eta_{n,k}  \geq 0$
 we have by  (\ref{eq:lem8-1}) that $F_k(\m_n)= \eta_{n,k}$ so that $\vert F_k(\m_n)- m_k\vert] = \vert F_k(\m_n)- m_k\vert$. If $\eta_{n,k} < 0$ we have by  (\ref{eq:lem8-2}) that $\eta_{n,k} < F_k(\mu_n)$ and hence $\vert \eta_{n,k}- m_k\vert =  m_k - \eta_{n,k} >  m_k - F_k(\mu_n) = \vert m_k - F_k(\mu_n)  \vert $, where the last equality is due to (\ref{eq:lem8-0}).
 
 \end{proof}
\section{Estimating Uniform Finite Sample Stickiness}
Denote  the standard-normal-cdf by $\Phi(x) = \frac{1}{\sqrt{2 \pi}} \int_{- \infty}^x \exp(-t^2/2)dt$.

\begin{Th}\label{berryesseentheorem}
\textbf{(Berry-Esseen, \citet[p. 42]{esseen_fourier_1945}, \cite{shevtsova2011})}
Let $Z_1,Z_2,...,Z_n$ be iid random variables on $\R$ with mean zero and finite third moment $\E\left[ \vert Z_1\vert^3 \right] < \infty$. Denote by $\sigma^2$ the variance and by $\hat F_n(z)$ the cumulative distribution function of the random variable
\begin{align*}
\frac{Z_1+Z_2 + ... + Z_n}{\sqrt{n}\sigma}.
\end{align*}
Then there is a finite positive constant $C_S\leq 0.4748$ such that
\begin{align} \label{berryesseeninequality}
\vert \hat F_n(z) - \Phi(z) \vert \le  \frac{\E\left[ \vert Z_1\vert^3 \right]}{\sqrt{n} \sigma^3}C_S.
\end{align}
\end{Th}
%\begin{proof}
%See \cite{esseen_fourier_1945}, page 42.
%\end{proof}
%\begin{Rm}
% \cite{shevtsova2011} proved that the Berry–Esseen inequality holds with
%$C_S = 0.4748.$
%\end{Rm}

For all of the following let $X_1,\ldots,X_n\iid X$ be nondegenerate random variables on $\cS_K$ with Fr\'echet mean $\m \in L_k$ for some $k\in \{1,\ldots,K\}$ and
$
\E \left[ \dis^3(X,\mathbf 0) \right] < \infty$. 
Let
\begin{align*}
    p_n &= \sum_{i = 1}^K \Phi\left( \frac{\sqrt{n} m_i}{\sigma_i} \right) + \sum_{i = 1}^K \frac{\mathbb{E} \left[ \vert F_i(X)-m_i\vert|^3\right]}{\sqrt{n} \sigma_i^3} C_S\,, \\
p_{n,k} &= \Phi\left( \frac{\sqrt{n} m_k}{\sigma_k} \right) - \frac{\mathbb{E} \left[ \vert F_k(X)-m_k\vert|^3\right]}{\sqrt{n} \sigma_k^3} C_S.
\end{align*}
%\begin{align*}
%    \tilde p_n &= \sum_{i = 1}^K \Phi\left( \frac{\sqrt{n} m_i}{\sigma_i} \right) + \sum_{i = 1}^K \frac{\mathbb{E} \left[ d^3(F_i(X),m_i)\right]}{\sqrt{n} \sigma_i^3} C_S\,, \\
%        \tilde p_{n,k} &= \sum\limits_{\substack{l \neq i,k\\ 1 \le l \le K}} \Phi\left( \frac{\sqrt{n} m_i}{\sigma_i} \right) + \sum\limits_{\substack{l \neq i,k\\ 1 \le l \le K}} \frac{\mathbb{E} \left[ d^3(F_i(X),m_i)\right]}{\sqrt{n} \sigma_i^3} C_S\,, \\
%p_{n,k} &= \Phi\left( \frac{\sqrt{n m_k^2+1}}{\sigma_k} \right) - \frac{\mathbb{E} \left[ d^3(F_k(X),m_k)\right]}{\sqrt{n} \sigma_k^3} C_S\,,
%\end{align*}
For $i\in \{1,\ldots,K\}$ set
\begin{align*}
    A_i = \{ \eta_{n,i} = F_i(\m_n)\}\,,\quad A = A_1\cup \ldots \cup A_K\,.
\end{align*}
Due nondegeneracy, (\ref{eq:eta}) implies that the $A_i$ ($i=1,\ldots,K$) are disjoint, see also \cite[Theorem 2.9]{HHMMN13}.
\begin{Lem} \label{tilde_p_n}
For all $n \in \mathbb{N}$, $\mathbb{P} \left( A_k \right) \geq p_{n,k}$ and $\mathbb{P} \left( A \right) \leq {p}_n. $
\end{Lem}

\begin{proof}
Fix $i \in \{1,..., K\}$. By Lemma \ref{folded_average_related_to_folded_sample_mean}, 
\end{proof}
\begin{align*}
\mathbb{P}(A_i) =
\mathbb{P}\left\{ \eta_{n,i} \geq 0 \right\}
&= \mathbb{P}\left\{\frac{\sqrt{n}(- \eta_{n,i}+m_i)}{\sigma_i} \le \frac{\sqrt{n}m_i}{\sigma_i} \right\}.
\end{align*}
Setting $Z_j = F_i(X_j) -m_i$ ($j=1,\ldots,n$) and $z=\frac{\sqrt{n}m_i}{\sigma_i}$ we can apply the Berry-Esseen theorem, Theorem \ref{berryesseentheorem}, yielding for $i =k$ 
\begin{align*}
    \mathbb{P} \left( A_k \right) & \geq  \Phi\left( \frac{\sqrt{n} m_k}{\sigma_k} \right) - \frac{\mathbb{E} \left[ \vert F_k(X)-m_k\vert|^3\right]}{\sqrt{n} \sigma_k^3} C_S = p_{n,k}
\end{align*}
and 
\begin{align*}
\mathbb{P}(A)  
 &\le \sum_{i=1}^K\Phi\left(\frac{\sqrt{n}m_i}{\sigma_i}\right)+\sum_{i=1}^K \frac{\E\left[ \vert F_i(X)- m_i\vert^3 \right]}{\sqrt{n} \sigma_i^3}\,C_S = p_n\,. %0.4748. 
%\mathbb{P}\left\{ \eta_{n,i} \geq 0 \right\} 
% &\le \Phi\left(\frac{\sqrt{n}m_i}{\sigma_i}\right)+\frac{\E\left[ \vert F_i(X)- m_i\vert^3 \right]}{\sqrt{n} \sigma_i^3}\,C_S\,. %0.4748. 
\end{align*}

\begin{Th}\label{theorem_conditions_fss_2}
A nondegenerate random variable on $X$  on $\cS_K$ with mean $\mu \in L_k$ and second folded moment $\sigma_k^2$, $k\in \{1,\ldots,K\}$, is finite sample sticky of level 
    $$\rho = \min_{n \in \{N,N+1,...,N^l\}} 1- p_n - \frac{n m_k^2}{\sigma_k^2}(1- p_{n,k})\, $$
 with scale $l$ and basis $N$, if there is $l \in \N_{\geq 2}$ such that $ p_n < 1 $ and $p_{n,k} \geq 0$, and if 
 $$p_n + \frac{n m_k^2}{\sigma_k^2}(1- p_{n,k}) < 1\, $$
 for all $n \in \{ N,N+1, ..., N^l\}$.
 \end{Th}

\begin{proof}
By the law of total expectation, Corollary \ref{cor:folded_average} and Lemma  \ref{tilde_p_n}, exploiting that $A^c \subseteq A_k^c$, we obtain
\begin{align*}
    \md &= \EE[\md] = \frac{n}{\sigma^2_k}
    \E\left[ \E [\dis^2(\mu_n,\mu) \mid A] \mathbb{P}(A) + \E [\dis^2(\mu_n,\mu) \mid A^c]\mathbb{P}(A^c) \right]\, \\
    &= \frac{n}{\sigma^2_k}
    \E [\dis^2(\mu_n,\mu)] \mathbb{P}(A) +
    \frac{n m_k^2}{\sigma^2_k}\mathbb{P}(A^c)\, \\
    &\le p_n + \frac{n m_k^2}{\sigma^2_k}(1-p_{n,k}).
\end{align*}
\end{proof}

\section{Example and Simulations}
\begin{Ex}
For $t> 0$ let $X_t$ be a random variable on $\Sp_K$ with probabilities 
\begin{align*}
\mathbb{P} \left\lbrace X_t = (K - 1 + Kt,K) \right\rbrace = \frac{1}{K} = \mathbb{P} \left\lbrace X_t = (1,i) \right\rbrace,  \qquad k = i,...,K-1\,,
\end{align*}
%For $i \neq K$ the probability of the folded random variable is given by 
%\begin{align*}
%\mathbb{P} \left\lbrace F_i(X_t) = -(K - 1 + Kt) \right\rbrace = \frac{1}{K} = \mathbb{P} \left\lbrace F_i(X_t) = 1 \right\rbrace, & \\\mathbb{P} \left\lbrace F_i(X_t) = -1 \right\rbrace = \frac{K-2}{K}, &  \qquad k = 1,...,K-1;
%\end{align*}
and first moments
\begin{align*}
m_i = \frac{4-K(2+t)}{K}, \qquad m_K = t.
\end{align*}
%where $m_K$ is the expected value of $F_K(X_t)$ defined by
%\begin{align*}
%\mathbb{P} \left\lbrace F_K(X_t) = (K - 1 + Kt) \right\rbrace = \frac{1}{K}, \qquad \mathbb{P} \left\lbrace F_K(X_t) = -1 \right\rbrace = \frac{K-1}{K}
%\end{align*}
Hence $X_t$ is nonsticky with $\mu \in L_K$ for $t > 0$  and $$\mathbb{E}[\dis^3(X_t,0)] = \frac{K-1}{K} + \frac{(K-1+Kt)^3}{K}< \infty.$$
We can therefore apply Theorem \ref{theorem_conditions_fss_2}. Figure \ref{picturestheorem2} illustrates intervals of sample sizes displaying finite sample stickiness for the 3-Spider and $ t \in \{10^{-2},10^{-3}, 10^{-4}\}$. The explicit bound for the modulation derived by Theorem \ref{theorem_conditions_fss_2} is given as an orange dashed line.
\end{Ex}

\vspace*{-0.5cm}

\begin{figure}[h!]
     \centering
     \begin{subfigure}[h]{0.32\textwidth}
         \includegraphics[width=\textwidth]{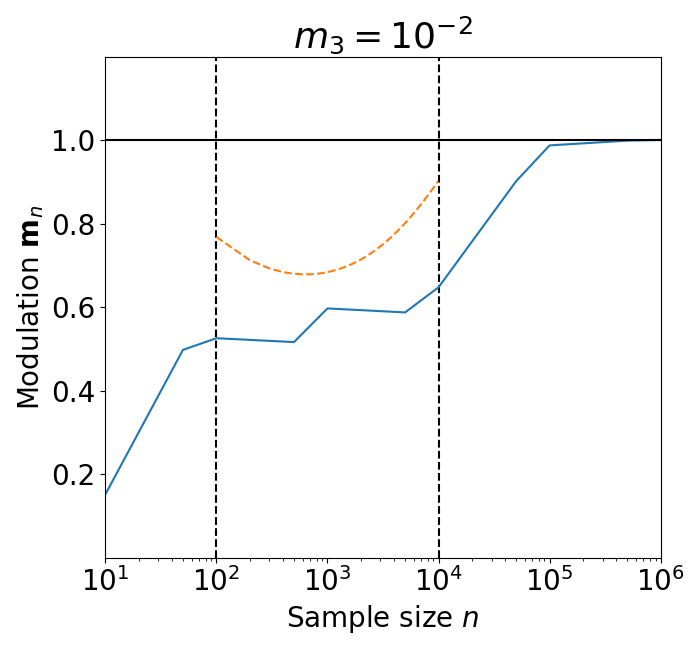}
         \caption{Finite sample stickiness of level $\rho = 0.1$ with scale $2$ and base $N = 100$.}
     \end{subfigure}
     \hfill
     \begin{subfigure}[h]{0.32\textwidth}
         \includegraphics[width=\textwidth]{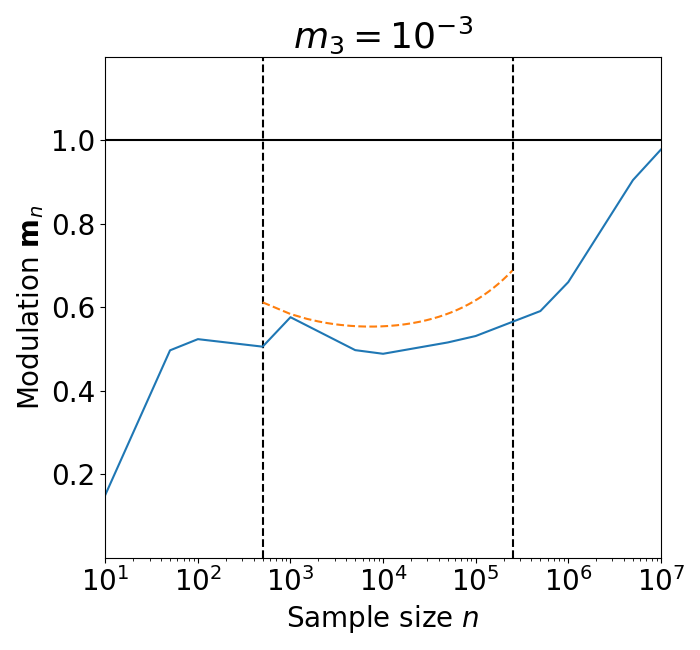}
         \caption{Finite sample stickiness of level $\rho = 0.31$ with scale $2$ and base $N = 500$.}

     \end{subfigure}
     \hfill
          \begin{subfigure}[h]{0.32\textwidth}
         \includegraphics[width=\textwidth]{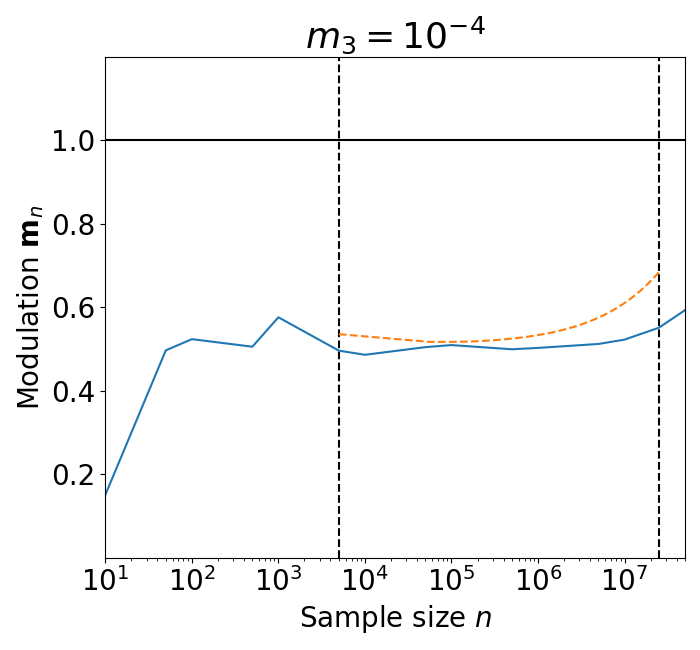}
         \caption{Finite sample stickiness of level $\rho = 0.31$ with scale $2$ and base $N = 5000$.}
     \end{subfigure}
        \caption{}\label{picturestheorem2}

\end{figure}

% \vspace*{-0.5cm}

\section*{Acknowledgment}
All authors acknowledge DFG HU 1575/7, DFG GK 2088, DFG SFB 1465 and the Niedersachsen Vorab of the Volkswagen Foundation. The work was done partially while the 2nd author was participating in the program of the Institute for Mathematical Sciences, National University of Singapore, in 2022.

\DirectBibliography{shape.bib,trees.bib}

\begin{thebibliography}{}

\bibitem[\protect\citeauthoryear{Barden, Le, and Owen}{Barden
  et~al.}{2013}]{BardenLeOwen2013}
Barden, D., H.~Le, and M.~Owen (2013).
\newblock Central limit theorems for {F}r{\'e}chet means in the space of
  phylogenetic trees.
\newblock {\em Electron. J. Probab\/}~{\em 18\/}(25), 1--25.

\bibitem[\protect\citeauthoryear{Barden, Le, and Owen}{Barden
  et~al.}{2018}]{Barden2018limiting}
Barden, D., H.~Le, and M.~Owen (2018).
\newblock Limiting behaviour of {F}r{\'e}chet means in the space of
  phylogenetic trees.
\newblock {\em Annals of the Institute of Statistical Mathematics\/}~{\em
  70\/}(1), 99--129.

\bibitem[\protect\citeauthoryear{Billera, Holmes, and Vogtmann}{Billera
  et~al.}{2001}]{BilleraHolmesVogtmann2001}
Billera, L., S.~Holmes, and K.~Vogtmann (2001).
\newblock Geometry of the space of phylogenetic trees.
\newblock {\em Advances in Applied Mathematics\/}~{\em 27\/}(4), 733--767.

\bibitem[\protect\citeauthoryear{Eltzner}{Eltzner}{2022}]{eltzner2022geometrical}
Eltzner, B. (2022).
\newblock Geometrical smeariness--a new phenomenon of {F}r{\'e}chet means.
\newblock {\em Bernoulli\/}~{\em 28\/}(1), 239--254.

\bibitem[\protect\citeauthoryear{Eltzner, Hansen, Huckemann, and
  Sommer}{Eltzner et~al.}{2023}]{eltzner2023diffusion}
Eltzner, B., P.~Hansen, S.~F. Huckemann, and S.~Sommer (2023).
\newblock Diffusion means in geometric spaces.
\newblock {\em Bernoulli\/}.
\newblock to appear.

\bibitem[\protect\citeauthoryear{Eltzner and Huckemann}{Eltzner and
  Huckemann}{2019}]{EltznereHuckemann2019}
Eltzner, B. and S.~F. Huckemann (2019).
\newblock A smeary central limit theorem for manifolds with application to
  high-dimensional spheres.
\newblock {\em Ann. Statist.\/}~{\em 47\/}(6), 3360--3381.

\bibitem[\protect\citeauthoryear{Eltzner, Hundrieser, and Huckemann}{Eltzner
  et~al.}{2021}]{eltzner2021finite}
Eltzner, B., S.~Hundrieser, and S.~Huckemann (2021).
\newblock Finite sample smeariness on spheres.
\newblock In {\em Geometric Science of Information: 5th International
  Conference, GSI 2021, Paris, France, July 21--23, 2021, Proceedings 5}, pp.\
  12--19. Springer.

\bibitem[\protect\citeauthoryear{Esseen}{Esseen}{1945}]{esseen_fourier_1945}
Esseen, C.-G. (1945).
\newblock Fourier analysis of distribution functions. a mathematical study of
  the {L}aplace-{G}aussian law.
\newblock {\em Acta Mathematica\/}~{\em 77}, 1--125.

\bibitem[\protect\citeauthoryear{Fr\'echet}{Fr\'echet}{1948}]{F48}
Fr\'echet, M. (1948).
\newblock Les \'el\'ements al\'eatoires de nature quelconque dans un espace
  distanci\'e.
\newblock {\em Annales de l'Institut de Henri Poincar\'e\/}~{\em 10\/}(4),
  215--310.

\bibitem[\protect\citeauthoryear{Hotz and Huckemann}{Hotz and
  Huckemann}{2015}]{HH15}
Hotz, T. and S.~Huckemann (2015).
\newblock Intrinsic means on the circle: Uniqueness, locus and asymptotics.
\newblock {\em Annals of the Institute of Statistical Mathematics\/}~{\em
  67\/}(1), 177--193.

\bibitem[\protect\citeauthoryear{Hotz, Huckemann, Le, Marron, Mattingly,
  Miller, Nolen, Owen, Patrangenaru, and Skwerer}{Hotz et~al.}{2013}]{HHMMN13}
Hotz, T., S.~Huckemann, H.~Le, J.~S. Marron, J.~Mattingly, E.~Miller, J.~Nolen,
  M.~Owen, V.~Patrangenaru, and S.~Skwerer (2013).
\newblock Sticky central limit theorems on open books.
\newblock {\em Annals of Applied Probability\/}~{\em 23\/}(6), 2238--2258.

\bibitem[\protect\citeauthoryear{Huckemann, Mattingly, Miller, and
  Nolen}{Huckemann et~al.}{2015}]{H_Mattingly_Miller_Nolen2015}
Huckemann, S., J.~C. Mattingly, E.~Miller, and J.~Nolen (2015).
\newblock Sticky central limit theorems at isolated hyperbolic planar
  singularities.
\newblock {\em Electronic Journal of Probability\/}~{\em 20\/}(78), 1--34.

\bibitem[\protect\citeauthoryear{Huckemann and Eltzner}{Huckemann and
  Eltzner}{2020}]{huckemann2020data}
Huckemann, S.~F. and B.~Eltzner (2020).
\newblock Data analysis on nonstandard spaces.
\newblock {\em Wiley Interdisciplinary Reviews: Computational Statistics\/},
  e1526.

\bibitem[\protect\citeauthoryear{Huckemann and Eltzner}{Huckemann and
  Eltzner}{2024}]{HE24Foundations}
Huckemann, S.~F. and B.~Eltzner (2024).
\newblock {\em Foundations of Non-Euclidean Statistics}.
\newblock London: Chapman \& Hall/CRC Press.
\newblock in preparation.

\bibitem[\protect\citeauthoryear{Hundrieser, Eltzner, and Huckemann}{Hundrieser
  et~al.}{2020}]{HundrieserEltznerHuckemann2020}
Hundrieser, S., B.~Eltzner, and S.~F. Huckemann (2020).
\newblock Finite sample smeariness of {F}r\'echet means and application to
  climate.
\newblock {\em arXiv preprint arXiv:2005.02321\/}.

\bibitem[\protect\citeauthoryear{Pennec}{Pennec}{2019}]{Pennec2019}
Pennec, X. (2019).
\newblock Curvature effects on the empirical mean in {R}iemannian and affine
  manifolds: a non-asymptotic high concentration expansion in the small-sample
  regime.
\newblock {\em arXiv preprint arXiv:1906.07418\/}.

\bibitem[\protect\citeauthoryear{Shevtsova}{Shevtsova}{2011}]{shevtsova2011}
Shevtsova, I. (2011).
\newblock On the absolute constants in the {B}erry-{E}sseen type inequalities
  for identically distributed summands.

\bibitem[\protect\citeauthoryear{Sturm}{Sturm}{2003}]{Sturm2003}
Sturm, K. (2003).
\newblock Probability measures on metric spaces of nonpositive curvature.
\newblock {\em Contemporary mathematics\/}~{\em 338}, 357--390.

\bibitem[\protect\citeauthoryear{Tran, Eltzner, and Huckemann}{Tran
  et~al.}{2021}]{DHH-GSI21}
Tran, D., B.~Eltzner, and S.~F. Huckemann (2021).
\newblock Smeariness begets finite sample smeariness.
\newblock In {\em Geometric Science of Information 2021 proceedings}, pp.\
  29--36. Springer.

\end{thebibliography}
\end{document}